\documentclass[letterpaper,12pt]{article}
\usepackage[utf8]{inputenc}
\usepackage[margin=1in]{geometry}
\usepackage{mathtools}
\usepackage{amsmath,amsthm}
\usepackage{placeins}
\usepackage{apacite}
\usepackage{rotating}
\usepackage{amssymb}
\usepackage{color}
\usepackage{natbib}
\usepackage[pdftex,pagebackref=false]{hyperref}
\hypersetup{
    plainpages=false,       % needed if Roman numbers in frontpages
    unicode=false,          % non-Latin characters in Acrobat’s bookmarks
    pdftoolbar=true,        % show Acrobat’s toolbar?
    pdfmenubar=true,        % show Acrobat’s menu?
    pdffitwindow=false,     % window fit to page when opened
    pdfstartview={FitH},    % fits the width of the page to the window
    pdfauthor={Chi-Kuang Yeh},    % author: CHANGE THIS TEXT! and uncomment this line
    pdfsubject={Statistics},  % subject: CHANGE THIS TEXT! and uncomment this line
    pdfnewwindow=true,      % links in new window
    colorlinks=true,        % false: boxed links; true: colored links
    linkcolor=[RGB]{198, 0, 120},
    citecolor=[RGB]{228, 180, 41},
    filecolor=magenta,      % color of file links
    urlcolor=cyan           % color of external links
}

\usepackage[english]{babel}
\newtheorem{theorem}{Theorem}
\usepackage{bm}
\usepackage{chemfig}
\usepackage{algorithmic}
\usepackage[ruled]{algorithm2e} % Algorithms
\usepackage{float}
\usepackage{bbold}
\usepackage{subfig}

\usepackage[parfill]{parskip}

\title{\bf Estimating Boltzmann Averages for Protein Structural Quantities Using Sequential Monte Carlo}
\author{Zhaoran Hou and Samuel W.K. Wong\footnote{Author for correspondence: samuel.wong@uwaterloo.ca} \\ University of Waterloo}
\date{October 18, 2022}                                           % Activate to display a given date or no date

\begin{document}
\maketitle

\begin{abstract}
Sequential Monte Carlo (SMC) methods are widely used to draw samples from intractable target distributions. Particle degeneracy can hinder the use of SMC when the target distribution is highly constrained or multimodal. As a motivating application, we consider the problem of sampling protein structures from the Boltzmann distribution. This paper proposes a general SMC method that propagates multiple descendants for each particle, followed by resampling to maintain the desired number of particles. Simulation studies demonstrate the efficacy of the method for tackling the protein sampling problem. As a real data example, we use our method to estimate the number of atomic contacts for a key segment of the SARS-CoV-2 viral spike protein.

\vspace{9pt}
\par
\end{abstract}\par

\def\thefigure{\arabic{figure}}
\def\thetable{\arabic{table}}

\renewcommand{\theequation}{\thesection.\arabic{equation}}

\fontsize{12}{14pt plus.8pt minus .6pt}\selectfont

\section{Introduction}\label{section: introduction}

Sequential Monte Carlo (SMC) methods, also known as particle filters, are simulation-based Monte Carlo algorithms for sampling from a target distribution.  SMC originated from on-line inference problems in dynamical systems, where observations arrive sequentially and interest lies in the posterior distribution of hidden state variables \citep{liu1998sequential}; since then, SMC methods have been used to solve a wide range of practical problems.   %The key idea of SMC is to represent the target distribution with a set of weighted particles that are updated incrementally \citep{del1997nonlinear}, which has been applied to problems including target tracking \citep{hue2002sequential, moghaddasi2020hybrid}, time series analysis \citep{urteaga2017sequential,ping2018particle}, and computer vision \citep{vermaak2001sequential}. Moreover, SMC has increasingly found utility in settings beyond dynamical systems, such as graphical models \citep{paige2016inference}, population genetics \citep{smith2017infectious}, and graph matching \citep{jun2019sequential}.

This paper proposes a SMC method that generalizes the propagation and resampling strategy of \cite{fearnhead2003line}, by constructing a sequence of upsampling and downsampling steps as we shall subsequently define. Our method is motivated by the problem of sampling 3-D structures of proteins from the Boltzmann distribution, which is a challenging task due to atomic interactions and constraints (see Section \ref{section: motivation} for scientific background). These constraints can exacerbate particle degeneracy when using SMC to obtain samples, for which our method provides a solution.

We begin with a review of the relevant SMC concepts following \cite{doucet2001introduction}.
Assume we have random variables $(\mathbf{x}_0, \dots, \mathbf{x}_T)$ (denoted by $\mathbf{x}_{0: T}$)  with continuous support $\mathcal{X}^{(T+1)}$, and we wish to draw samples from the target distribution  $p(\mathbf{x}_{0: T})$. Let $f_{T}: \mathcal{X}^{(T+1)} \rightarrow \mathbb{R}^{n_{f_{T}}}$ denote a square integrable function of interest, then its expectation with respect to $p(\mathbf{x}_{0: T})$ is given by
\begin{eqnarray}
E_p\left[f_{T}\left(\mathbf{x}_{0: T}\right)\right] = \int f_{T}\left(\mathbf{x}_{0: T}\right) p\left(\mathbf{x}_{0: T}\right) d \mathbf{x}_{0: T}. \label{eq:smcintegral}
\end{eqnarray}
Since this integration is usually analytically intractable, the goal is to produce a set of particles $\{\mathbf{x}_{0: T}^{(n)}\}_{n=1}^N$ with weights $\{w(\mathbf{x}_{0: T}^{(n)})\}_{n=1}^N$ such that an estimate to the integral is given by
\begin{equation}
E_p\left[f_{T}\left(\mathbf{x}_{0: T}\right)\right] %\int f_{T}\left(\mathbf{x}_{0: T}\right) p\left(x_{0: T} \mid y_{1: T}\right) d x_{0: T} 
\approx \frac{\sum_{n=1}^N f_{T}(\mathbf{x}_{0: T}^{(n)})w(\mathbf{x}_{0: T}^{(n)})}{\sum_{n=1}^N w(\mathbf{x}_{0: T}^{(n)})},
\label{eq:SMCapproximation}
\end{equation}
which requires the weights to be proper with respect to $p(\mathbf{x}_{0: T})$ \citep{liu1998sequential,liu2001monte}, i.e., $$E\left[f_{T}(\mathbf{x}_{0: T}^{(n)})w(\mathbf{x}_{0: T}^{(n)})\right] \propto E_p\left[f_{T}(\mathbf{x}_{0: T})\right].$$ Often, $p(\mathbf{x}_{0:T})$ does not adopt a form from which we can directly sample (or use importance sampling efficiently) in practice (e.g.,  \cite{jacquier2002bayesian,carvalho2010particle}). In this case, a set of auxiliary distributions $\{p_t(\mathbf{x}_{0: t})\}_{t=0}^T$ can be introduced, % with $p_t: \mathcal{\mathcal{X}}^{(t+1)} \rightarrow \mathbb{R}$ and 
with $p_T(\mathbf{x}_{0:T}) = p(\mathbf{x}_{0:T})$, to facilitate sequential sampling \citep{liu2001monte}; note that $p_t(\mathbf{x}_{0: t})$ does not need to equal $\int p(\mathbf{x}_{0: T})d\mathbf{x}_{t+1: T}$ when $t<T$. To then construct $\{\mathbf{x}_{0:T}^{(n)}\}_{n=1}^N$, SMC generates particles according to the auxiliary distributions via a sequence of propagation and resampling steps. Assume a set of weighted particles $\{\mathbf{x}_{0:t-1}^{(n)}\}_{n=1}^N$ have been sampled from $p_{t-1}(\mathbf{x}_{0: t-1})$ for $0 < t \leq T$; for each existing  $\mathbf{x}_{0:t-1}^{(n)}$, the propagation step samples $\mathbf{x}_{t}$ and appends it to the existing particle to form  $(\mathbf{x}_{0:t-1}^{(n)}, \mathbf{x}_{t})$ as a weighted sample from $p_{t}(\mathbf{x}_{0:t})$. Afterwards, a resampling step may be done to preserve a set of more evenly-weighted particles. Distinct particles with more evenly-distributed weights are desired to better represent the target distribution and thus reduce the Monte Carlo variance of the estimates in Equation \eqref{eq:SMCapproximation}.

Sequential importance sampling with resampling (SISR) is a common framework to implement propagation with the help of importance distributions $\eta(\mathbf{x}_0), \eta(\mathbf{x}_1\mid \mathbf{x}_0),\dots, \eta(\mathbf{x}_T\mid \mathbf{x}_{0:T-1})$ and resampling \citep{liu1995blind, liu1998sequential}, as summarized in Algorithm \ref{SISR}. %The SISR framework encompasses some earlier methods such as the bootstrap filter (BF) of \cite{gordon1993novel} which is considered as the first particle filter. %The BF assumes sampling from a Markovian process and adapts Step 1 and 2 of Algorithm \ref{SISR} as $\eta(\mathbf{x}_t\mid \mathbf{x}_{0:t-1})=\eta(\mathbf{x}_t\mid \mathbf{x}_{t-1})$ and $p_t({\mathbf{x}}_{t}\mid \mathbf{x}_{0:t-1})=p_t({\mathbf{x}}_{t}\mid \mathbf{x}_{t-1})$.
To briefly note some key features of SISR, Step 1 samples one descendant for each particle and Step 3 resamples from the propagated particles. If Step 3 is omitted, the SISR framework reduces to sequential importance sampling (SIS). 
The necessity of Step 3 depends on the importance weights: if the importance weights are all constant, resampling only reduces the distinction of the particles and thus increases Monte Carlo variance. However, the importance weights are usually uneven in practice; then without Step 3, some of the importance weights evolving in Step 2 may decay to zero along with the propagation, which is known as particle degeneracy. 

Many resampling schemes for Step 3 have been proposed to tackle particle degeneracy and maintain more even weights. \cite{gordon1993novel} adopted multinomial sampling with i.i.d.~draws; \cite{kitagawa1996monte} proposed stratified resampling which lines up the importance weights, divides the interval into equal parts and uniformly samples from each subinterval; \cite{liu1998sequential} proposed residual resampling, which implements random sampling after retaining copies of current particles based on the weights. \cite{li2022stratification} summarized some previous resampling schemes, including the three aforementioned ones, and showed the equivalency of optimal transport resampling and stratified resampling along with their optimality in one dimensional cases. The resampling schemes update the sample weights while preserving the proper weighting condition, but differ in the resulting Monte Carlo variance.
% needed for estimating the Monte Carlo integral, 
%but also introduce varying levels of randomness that affect the magnitude of the resulting Monte Carlo variance. %Thus besides maintaining proper weighting, resampling 
Thus, schemes that minimize the resampling variance are preferred.
Further, we note that these resampling schemes do not thoroughly solve the particle degeneracy of SISR for all situations. For example as seen later in the main application of this paper, the particle degeneracy encountered when sampling protein backbone segments makes SISR inapplicable: the particle weights can decay to zero so rapidly that there could be no particles with positive weights after just a few propagation and resampling steps. To deal with this issue naively, we would need to exponentially increase $N$ with the length of the protein segment to guarantee that SISR can successfully complete, which comes with an enormous computational burden.

\begin{algorithm}[t!]
\caption{Sequential importance sampling with resampling}\label{SISR}
Require: particle size $N$\;

Initialization: Sample $\{\mathbf{x}_{0}^{(n)}\}_{n=1}^N$ from $\eta(\mathbf{x}_0)$, and the weight $w(\mathbf{x}_{0}^{(n)}) \propto p_0(\mathbf{x}_0^{(n)})/\eta(\mathbf{x}_0^{(n)})$\;

\For{$t=1,\dots,T$}{
  Step 1: Sample $\widetilde{\mathbf{x}}_{t}$ from $\eta(\mathbf{x}_{t}\mid \mathbf{x}_{0:t-1}^{(n)})$ and set $\widetilde{\mathbf{x}}^{(n)}_{0:t} = (\mathbf{x}_{0:t-1}^{(n)}, \widetilde{\mathbf{x}}_{t})$ for each $n$\;
  Step 2: Evaluate the weight $w(\widetilde{\mathbf{x}}^{(n)}_{0:t}) \propto w(\mathbf{x}_{0:t-1}^{(n)})p_t(\widetilde{\mathbf{x}}^{(n)}_{0:t})/(p_{t-1}({\mathbf{x}}^{(n)}_{0:t-1})\eta(\widetilde{\mathbf{x}}_{t}\mid \mathbf{x}_{0:t-1}^{(n)}))$ for each $n$\;
  Step 3: Resample $N$ particles $\{\mathbf{x}_{0:{t}}^{(n)}\}_{n=1}^N$ from $\{\widetilde{\mathbf{x}}_{0:{t}}^{(n)}\}_{n=1}^N$ based on $\{w(\widetilde{\mathbf{x}}^{(n)}_{0:t})\}_{n=1}^N$ and update the weights $\{w(\mathbf{x}_{0:t}^{(n)})\}_{n=1}^N$\;
}
\end{algorithm}

\cite{fearnhead2003line} proposed an SMC algorithm for hidden Markov models restricted to a finite space $\mathcal{X}$ (i.e., $|\mathcal{X}| = M < \infty$), that circumvents some of the limitations of SISR. Their method explores every value of $\mathcal{X}$ and produces $M$ descendants during propagation for each of the $N$ particles. It then resamples $N$ distinct particles from the $MN$ particles such that the particle weights minimize the expected squared error loss function, and thus is an optimal resampling scheme.
%Therefore, it is an optimal resampling scheme in that it guarantees minimizing the variance of the Monte Carlo estimates.
This SMC algorithm has also been adopted in subsequent research due to its useful features. For example, \cite{zhang2007monte} investigated its application in sampling protein structures with a simplified discrete-state representation for the positions of the amino acids in a protein; in this context, the weight indicates the plausibility of a position under the given energy function. \cite{lin2008statistical} examined the effects of retaining the unfilled spaces of different shapes and sizes enclosed in the interior of proteins by adopting optimal resampling when considering two-dimensional lattice models of protein structures. \cite{fearnhead2007line} presented a variation on the optimal resampling by ordering the particles before resampling and illustrated the optimality of the extended resampling method in terms of minimizing the mean-square error for changepoint models. \cite{wong2018exploring} adopted the propagation idea with 
fixed $M=100$ (rather than exploring all possible values) to identify protein structures with low potential energy; their method omitted resampling and was primarily designed for structure prediction, so the produced samples are not properly weighted and cannot be used for Monte Carlo integration. 

We shall define \textit{upsampling} to be a propagation step that samples $MN$ descendants from $N$ particles ($M\ge1$), and \textit{downsampling} to be a resampling step which resamples $N$ particles from the $MN$ descendants. An \textit{upsampling-downsampling framework} %shall refer to any SMC algorithm that
combines these upsampling and downsampling features. The SMC method proposed in this paper, while motivated by the sampling problem in protein structures, is a generally applicable strategy for sampling from multivariate continuous distributions to compute Monte Carlo integrals. It may be especially effective when particle degeneracy cannot be solved by existing resampling schemes, e.g., for target distributions that are highly constrained or have many sharp local modes.

The remainder of the paper is laid out as follows. In Section \ref{section: motivation}, we introduce the scientific background of proteins and our goal of estimating Boltzman averages. In Section \ref{section: methodology}, we describe the construction of our SMC method and how to use it for sampling protein backbone segments. Section \ref{section: simus} presents two simulation studies on the performance of the SMC method: Section \ref{section: simu_one} investigates the role of the upsample size $M$; Section \ref{section: simu_two} illustrates the numerical convergence of our SMC estimates. These simulations also demonstrate the inefficacy of naive importance sampling and SISR for the protein sampling problem. In Section \ref{section: application}, we apply the proposed SMC method to estimate the number of atomic contacts for a key segment of the SARS-CoV-2 viral spike protein. In Section \ref{section: conclusion}, we briefly summarize the paper and its contributions, and discuss some potential future directions. Proofs of the theorems are provided in the Supplementary Materials.

\section{Motivating Application: Estimating Protein Structural Quantities}\label{section: motivation}

\subsection{Overview of protein structure}

Proteins have a crucial role in carrying out biological processes and their functions are dependent on their 3-D structures. A protein consists of a sequence of amino acids, where successive amino acids are connected by peptide bonds.  Four atoms including N, $\text{C}_\alpha$, C and O are common to each of the 20 different amino acid types and compose the backbone of the protein that can be visualized as 

\begin{center}
    \chemfig{\cdots--{C_\alpha}_t(-[2]R_t)--{C}_t(=[2]O_t)--N_{t+1} -- {C_\alpha}_{t+1}(-[2]R_{t+1})--\cdots}
\end{center}
for the amino acids with indices $t$ and $t+1$ in the amino acid sequence. Side chain groups extend from the C$_\alpha$ atoms (e.g., R$_t$ and R$_{t+1}$ in the visualized backbone structure) and distinguish 20 different types of amino acids.
Proteins generally adopt stable 3-D structures that are essentially determined by the sequence of the amino acids \citep{anfinsen1973principles}. At the same time, it is well-known that protein structures are not static; some dynamic movement is often observed %, i.e., their structures are not static and can change over time
\citep{bu2011proteins,fraser2009hidden,fraser2011accessing}. Thus, while the Protein Data Bank (PDB) \citep{berman2000protein} is the source of known protein structures from laboratory work, these should be considered as only static snapshots of the true structures.

An arrangement of the atoms of a protein in 3-D space is known as a \emph{conformation}. The positions (3-D coordinates) of the backbone atoms $a_t = (\text{C}_t, \text{O}_t, \text{N}_{t+1}, {\text{C}_\alpha}_{t+1})$ can be equivalently specified using dihedral angles $(\phi_t, \psi_t, \omega_t)$, where $\phi_t$ represents the dihedral angle of $\text{C}_{t-1} - \text{N}_t - \text{C}_\alpha{}_{t}$ and $\text{N}_t - \text{C}_\alpha{}_{t} - \text{C}_{t}$; $\psi_t$ represents the dihedral angle of $\text{N}_t - \text{C}_\alpha{}_{t} - \text{C}_{t}$ and $\text{C}_\alpha{}_{t} - \text{C}_{t} - \text{N}_{t+1}$; $\omega_t$ represents the dihedral angle of $\text{C}_\alpha{}_{t} - \text{C}_{t} - \text{N}_{t+1}$ and $\text{C}_{t} - \text{N}_{t+1} - \text{C}_\alpha{}_{t+1}$. Since bond lengths and bond angles are essentially fixed, $\phi_t$ governs the distance between C$_{t-1}$ and C$_{t}$; $\psi_t$ governs the distance between N$_{t}$ and N$_{t+1}$ and also determines the coordinates of O$_t$ which is connected to C$_t$ as part of the planar peptide bond. These peptide bonds connecting C$_{t}$ and N$_{t+1}$ are nearly non-rotatable, and thus the dihedral angle $\omega_t$ that governs the distance between ${\text{C}_\alpha}_{t}$ and ${\text{C}_\alpha}_{t+1}$ %and determines the coordinates of the ${\text{C}_\alpha}_{t+1}$ 
is usually close to 180°
\citep{esposito2005correlation}. In contrast, $\phi_t$ and $\psi_t$ can take a wide range of values over the continuous interval [-180°, 180°).  

In this paper, we focus on the 3-D structures of protein segments %(also known as loops) 
due to the important roles they can have in biological function, e.g., the highly dynamic region of the coronavirus spike protein binding with human cells (see real data example in Section \ref{section: application}). The structure prediction problem for segments, with the rest of the protein held fixed, has thus attracted much attention. The basic idea of computational structure prediction is to find the conformation with the lowest potential energy, which originates from %; this idea is based on 
the energy-landscape theory of protein folding \citep{onuchic1997theory}. Prediction methods for segments have included sequential sampling approaches \citep{tang2014fast, wong2018exploring}, robotics-inspired protocols  \citep{stein2013improvements}, and hybrid approaches that leverage databases of known structures \citep{marks2017sphinx}. These methods focus on predicting a single ``correct" conformation with the lowest potential energy. However, for exploring the full range of dynamic movement and estimating Boltzmann averages (according to a given energy function), new sampling approaches are needed.

We will show that the SMC method proposed in this paper is useful for tackling this sampling problem.
%Hold the rest of the target protein fixed, and
A backbone conformation for the segment of interest, consisting of amino acids $a_0, \ldots, a_{T}$, can be specified by the continuous state vector $\mathbf{x}_{0:T}$ where $\mathbf{x}_t = (\phi_{t}, \psi_{t}, \omega_t)$ for $t=0,\dots, T$. 
An energy function $H$ is used to evaluate the energies of conformations, and we assume the corresponding Boltzmann distribution characterizes their range of dynamic movement, 
\begin{eqnarray}
p(\mathbf{x}_{0:T}) \propto \exp \left\{-H(\mathbf{x}_{0:T})/\lambda\right\},  \label{pix}
\end{eqnarray}
where $H$ is a given energy function, and  $\lambda$ is the effective temperature which can be taken to be 1 by appropriately scaling $H$. Our SMC method is suitable for sampling conformations from Equation \eqref{pix} for estimating Boltzmann averages, i.e., we want to estimate $E[f(\mathbf{x}_{0:T})]$ with respect to the Boltzmann distribution, where $f$ is a protein structural quantity or statistic of interest  computed from a given conformation $\mathbf{x}_{0:T}$. 

\subsection{Protein structural quantities}\label{subsec:structquants}

In analyses of protein structures, various quantities or summary statistics are computed from conformations, depending on the intended application. The ${\text{C}_\alpha}$ backbone atoms often have an important role, e.g., contact prediction considers Euclidean distances among all pairs of ${\text{C}_\alpha}$'s  \citep{di2012deep,zheng2019deep}; the number of other ${\text{C}_\alpha}$ atoms within a sphere centered at an amino acid's ${\text{C}_\alpha}$ is counted for solvent surface estimation \citep{simons1997assembly}; the ${\text{C}_\alpha}$ atom is often chosen as a substitute for the entire amino acid in coarse-grain protein representations, to improve the computational efficiency of  structure and function prediction \citep{roy2010tasser}. Statistics involving \textit{atomic contacts}, i.e., the number of non-bonded atoms within a given radius of a selected atom, are also important:  \cite{karlin1999atom} investigated atom density in terms of the average number of neighboring C, O, and N atoms; \cite{pintar2002cx} introduced the CX algorithm to identify protruding atoms in terms of the number of non-hydrogen atoms within a given radius; \cite{mihel2008psaia} developed a software tool that integrates several algorithms, including CX, to compute protein quantities for protein interaction; \cite{ligeti2017cx} developed the Protein Core Workbench (PCW) server for explicit visualization of atomic contacts. Protruding atoms, i.e., those with fewer atomic contacts, tend to be on the surface of the protein and more frequently involved in binding activities \citep{pintar2002cx}.

Therefore, to illustrate our methodology we consider similar statistics that summarize aspects of protein segment conformations: (i) distances between pairs of ${\text{C}_\alpha}$ atoms, as used in contact prediction; (ii) number of atomic contacts of ${\text{C}_\alpha}$ atoms, within a radius of 7 \AA.

\section{Methodology}\label{section: methodology}
In this section we summarize the SMC strategy of \cite{fearnhead2003line}, followed by its generalization in the form of our proposed SMC method for sampling from multivariate continuous distributions. We then present the method as applied to sampling protein segment conformations from the Boltzmann distribution for a given energy function $H$.

\subsection{Review of Fearnhead and Clifford's SMC method}

We review the relevant details of  \cite{fearnhead2003line} in this subsection. Consider a hidden Markov model setup with hidden process $\{\mathbf{x}_t\}_{t=0}^T$, each with finite state space $\mathcal{X}$ such that $|\mathcal{X}| = M < \infty$, observations $\{y_t\}_{t=1}^T$, transition probability $q(\mathbf{x}_t\mid \mathbf{x}_{t-1})$ and observation probability $l(y_t \mid \mathbf{x}_t)$ for $0<t\leq T$. Given a set of weighted particles $\{\mathbf{x}_{0:t-1}^{(n)}\}_{n=1}^N$ up to index $t-1$, each $\mathbf{x}_{0:t-1}^{(n)}$ produces $M$ distinct descendants, denoted by $\mathbf{x}_{t}^{(n,m)}$, one for each value of $\mathcal{X}$. The weight of the propagated particle $\mathbf{x}_{0:t}^{(n,m)} = (\mathbf{x}_{0:t-1}^{(n)}, \mathbf{x}_{t}^{(n,m)})$ satisfies
$$
w(\mathbf{x}_{0:t}^{(n,m)}) \propto w(\mathbf{x}_{0:t-1}^{(n)}) q(\mathbf{x}_t^{(n,m)}\mid \mathbf{x}_{0:t-1}^{(n)}) l(y_t\mid \mathbf{x}_t^{(n,m)}).
$$ To resample $N$ particles from the $MN$ candidates, the constant $c_t$ in
\begin{equation}
    \sum_{n=1}^N\sum_{m=1}^M \min(c_tw(\mathbf{x}_{0:t}^{(n,m)}), 1) = N \label{eq:threshold}
\end{equation}
is solved. Let $L$ be the number of particles with weights greater than or equal to $1/c_t$; these $L$ particles are all preserved, and stratified sampling \citep{carpenter1999improved} is used to resample another $N-L$ particles without replacement. \cite{fearnhead2003line} showed that this downsampling scheme minimizes
$$
E\left[\sum_{n=1}^N\sum_{m=1}^M(Q(\mathbf{x}_{0:t}^{(n,m)})-\gamma_{t}^{(n,m)})^2\right]
$$
where $\gamma_{t}^{(n,m)}$ denotes the weight with respect to $p_t$ and $Q(\mathbf{x}_{0:t}^{(n,m)})$ is the %assigned particle 
stochastic weight of the sampled particle, and thus is optimal among downsampling schemes.

\subsection{A general upsampling-downsampling SMC framework} \label{sec:proposed:methodology}
Our goal is to sample from a general multivariate distribution $p(\mathbf{x}_{0:T})$ where $\mathbf{x}_t$ is a random vector with continuous support, with the help of a proposed upsampling-downsampling SMC framework. %spaces (i.e., $|\mathcal{X}| = \mathfrak{c}$). 
In this situation, it is impossible to explore every value of $\mathcal{X}$ as in \cite{fearnhead2003line}, so we use importance distributions $\eta(\mathbf{x}_0), \eta(\mathbf{x}_1\mid \mathbf{x}_0),\dots, \eta(\mathbf{x}_T\mid \mathbf{x}_{0:T-1})$ to facilitate sampling.

Assume for a prechosen particle size $N$, upsample size $M$, and $t>0$, we have a set of generated particles $\{\mathbf{x}_{0:t-1}^{(n)}\}_{n=1}^N$ with %the downsampled
weights $\{w(\mathbf{x}_{0:t-1}^{(n)})\}_{n=1}^N$. For each $\mathbf{x}_{0:t-1}^{(n)}$, we sample $M$ descendants $\{\mathbf{x}_{t}^{(n,m)}\}_{m=1}^M$ from $\eta(\mathbf{x}_{t} \mid \mathbf{x}_{0:t-1}^{(n)})$ and define $\mathbf{x}_{0: t}^{(n, m)}=(\mathbf{x}_{0:t-1}^{(n)}, \mathbf{x}_{t}^{(n,m)})$ with
\begin{equation}
    w(\mathbf{x}_{0:t}^{(n,m)}) \propto w(\mathbf{x}_{0:t-1}^{(n)})  \frac{p_t(\mathbf{x}_{0: t}^{(n, m)})}{p_{t-1}(\mathbf{x}_{0: t-1}^{(n)})\eta(\mathbf{x}_{t}^{(n,m)} \mid \mathbf{x}_{0:t-1}^{(n)})}.
    \label{eq:upsampled_weight}
\end{equation}
These $MN$ propagated particles with the upsampled weights in Equation \ref{eq:upsampled_weight} are then carried to the downsampling step to resample $N$ particles.

The downsampling step resamples $N$ particles from the upsampled particles $\{\mathbf{x}_{0:t}^{(m,n)}, n=1,\dots,N \text{ and } m=1,\dots,M\}$, which we shall denote by $\{\mathbf{x}_{0:t}^{(n)}\}_{n=1}^N$. There are two possible cases after each upsampling step: (i) at least $N$ of $w(\mathbf{x}_{0:t}^{(n,m)})$'s are positive and (ii) less than $N$ of $w(\mathbf{x}_{0:t}^{(n,m)})$'s are positive. Note that \cite{fearnhead2003line} do not consider case (ii) as the Gaussian observation likelihood in their application always produces positive weights after upsampling. Handling case (ii) ensures a valid algorithm when particle degeneracy is very high, e.g., protein structure sampling could lead to most particles having zero weights after an upsampling step.
Compared to case (i), case (ii) is rare in practice but may occur when choices for $M$ and $N$ are too extreme (see simulation study in Section \ref{section: simu_one}).

For case (i), the downsampling step follows the method of \cite{fearnhead2003line}. After the $t$-th upsampling step, the threshold $c_t$ in the equation
$$
\sum_{n=1}^{N} \sum_{m=1}^{M} \min \left\{c_t w(\mathbf{x}_{0:t}^{(n,m)}), 1\right\}=N
$$
is solved. Let $L$ be the number of particles whose weights are greater than or equal to $1/c_t$; these are preserved together with the weights $w(\mathbf{x}_{0:t}^{(n,m)})$. $N-L$ particles are sampled without replacement from the remaining $MN-L$ particles proportional to their weights $w(\mathbf{x}_{0:t}^{(n,m)})$ and are assigned new weights $1/c_t$. Therefore, after the downsampling step, each $\mathbf{x}_{0:t}^{(n,m)}$ has a stochastic downsampled weight $Q(\mathbf{x}_{0:t}^{(n,m)})$ as
\begin{equation}
    Q(\mathbf{x}_{0:t}^{(n,m)})=\left\{\begin{array}{ll}w(\mathbf{x}_{0:t}^{(n,m)}) / q(\mathbf{x}_{0:t}^{(n,m)}) & \text { with probability } q(\mathbf{x}_{0:t}^{(n,m)}) \\ \\ 0 & \text { otherwise }\end{array}\right.
    \label{eq:downsampled_weight}
\end{equation}
where $q(\mathbf{x}_{0:t}^{(n,m)}) = \min\{c_tw(\mathbf{x}_{0:t}^{(n,m)}),1\}$.

This resampling scheme ensures that the downsampled weights $Q(\mathbf{x}_{0:t}^{(n,m)})$ are proper with respect to $p_t$ %(see Appendix A for proof)
and is optimal in terms of minimizing the expected squared error loss conditional on the upsampled particles. % (see Appendix B for proof). 
These results are formalized in Theorems \ref{thm:proper} and \ref{thm:optimal} below.
The $N$ downsampled particles are then the collection of those with realizations  $Q(\mathbf{x}_{0:t}^{(n,m)})  > 0$, which we set to be $\{\mathbf{x}_{0:t}^{(n)}\}_{n=1}^N$ with corresponding weights $\{w(\mathbf{x}_{0:t}^{(n)})\}_{n=1}^N$.

\begin{theorem} \label{thm:proper}
For \(t \in \{0,\dots,T\}\) and the auxiliary distribution \(p_t\), suppose there are at least \(N\) positive upsampled weights for all \(s \in \{0,\dots,t\}\). Let \(\mathbf{x}_{0:t}\) be any particle from the space \(\mathcal{X}^{t+1}\), then its downsampled weight \(Q(\mathbf{x}_{0:t})\) defined in Equation \eqref{eq:downsampled_weight} is proper with respect to \(p_t\). %See Appendix A for the proof.
\end{theorem}
\begin{proof}
See Section A in the Supplementary Materials.
\end{proof}

\begin{theorem} \label{thm:optimal}
For \(t \in \{0,\dots,T\}\) and upsampled particles \(\{\mathbf{x}_{0:t}^{(n,m)}; n=1,\dots,N, m=1,\dots,M\}\), suppose there are at least \(N\) positive upsampled weights at step \(t\). When each individual downsampled weight \(Q(\mathbf{x}_{0:t}^{(n,m)})\) is assigned by Equation \eqref{eq:downsampled_weight} with the constraint that no more than \(N\) elements of \(\{Q(\mathbf{x}_{0:t}^{(n,m)}); n=1,\dots,N, m=1,\dots,M\}\) are positive, then the downsampled weights minimize the conditional expected squared error loss
\[E\left[\sum_{n=1}^N\sum_{m=1}^M(Q(\mathbf{x}_{0:t}^{(n,m)})-\gamma_{t}^{(n,m)})^2 \middle\vert\ \{\mathbf{x}_{0:t}^{(n,m)}; n=1,\dots,N, m=1,\dots,M\}\right]\] where \(\gamma_{t}^{(n,m)} = \left(p_t(\mathbf{x}^{(n,m)}_{0:t})/\eta(\mathbf{x}^{(n,m)}_{0:t})\right)/\sum_{n=1}^N\sum_{m=1}^M\left(p_t(\mathbf{x}^{(n,m)}_{0:t})/\eta(\mathbf{x}^{(n,m)}_{0:t})\right)\). %See Appendix B for the proof.
\end{theorem}
\begin{proof}
See Section B in the Supplementary Materials.
\end{proof}

For case (ii), we have less than $N$ particles with positive weights so cannot preserve $N$ distinct particles via downsampling; instead, resampling $N$ particles with replacement is needed, similar to the resampling step in SISR. We use multinomial resampling for simplicity, which mimicks the resampling scheme of the bootstrap filter \citep{gordon1993novel}, but other schemes (see Introduction) can also be used. The key difference between case (i) and (ii) is that there are duplicates of particles with equal weights after downsampling for case (ii).

Finally to initialize the algorithm, $MN$ realizations of $\mathbf{x}_0$, 
i.e. $\{\mathbf{x}_{0}^{(n,m)},n=1,\dots N \text{ and } m=1,\dots,M\}$, are first sampled from $\eta(\mathbf{x}_0)$ (each with normalized importance weight proportional to $p_0(\mathbf{x}_0)/\eta(\mathbf{x}_0)$), followed by a downsampling step to obtain $N$ properly weighted particles representing $p_0(\mathbf{x}_0)$.
Algorithm \ref{our_SMC} summarizes our SMC method.

\begin{algorithm}[t!]
\caption{SMC with upsampling-downsampling framework}\label{our_SMC}
Require: particle size $N$, upsample size $M$\;

Initial upsampling: Sample $\{\mathbf{x}_{0}^{(n,m)},n=1,\dots,N \text{ and } m=1,\dots,M\}$ from $\eta(\mathbf{x}_0)$, each with the weight $w(\mathbf{x}_{0}^{(n,m)}) \propto p_0(\mathbf{x}_0^{(n,m)})/\eta(\mathbf{x}_0^{(n,m)})$\;

Initial downsampling: Resample $N$ particles, denoted by $\{\mathbf{x}_{0}^{(n)}\}_{n=1}^N$, from $\{{\mathbf{x}}^{(n,m)}_{0},n=1,\dots,N \text{ and } m=1,\dots M\}$ with weights $\{w(\mathbf{x}_{0}^{(n)})\}_{n=1}^N$ using the proposed scheme\;

\For{$t=1,\dots,T$}{
  Upsampling: Sample $\{{\mathbf{x}}^{(n,m)}_{t}\}_{m=1}^M$ from $\eta(\mathbf{x}_{t}\mid \mathbf{x}_{0:t-1}^{(n)})$ and set ${\mathbf{x}}^{(n,m)}_{0:t} = (\mathbf{x}_{0:t-1}^{(n)}, {\mathbf{x}}^{(n,m)}_{t})$ for each $n$, each with a weight $w({\mathbf{x}}^{(n,m)}_{0:t}) \propto w({\mathbf{x}}^{(n)}_{0:t-1})p_t(\mathbf{x}_{0: t}^{(n, m)})/(p_{t-1}(\mathbf{x}_{0: t-1}^{(n)})\eta(\mathbf{x}_{t}^{(n,m)} \mid \mathbf{x}_{0:t-1}^{(n)}))$\;
  Downsampling: Resample $N$ particles, denoted by $\{\mathbf{x}_{0:t}^{(n)}\}_{n=1}^N$, from $\{{\mathbf{x}}^{(n,m)}_{0:t},n=1,\dots,N \text{ and } m=1,\dots M\}$ with weights $\{w(\mathbf{x}_{0:t}^{(n)})\}_{n=1}^N$ using the proposed scheme\;
}
\end{algorithm}

\subsection{Sampling conformations from the Boltzmann distribution}\label{subsec:energy_function}

In the protein sampling context, we let $\mathbf{x}_t$ represent the dihedral angles $(\phi_{t}, \psi_{t}, \omega_t)$ with $\mathcal{X}=$ \mbox{(-180°, 180°$]^3$}. Since $\mathbf{x}_{0:T}$ corresponds directly to a backbone conformation for the segment of amino acids $a_0, \ldots, a_{T}$, the joint distribution of the dihedral vector $p(\mathbf{x}_{0:T})$ is the target distribution. The energy function $H$ in Equation \eqref{pix} is assumed to be given, consisting of three components as in \cite{wong2018exploring}:
\begin{itemize}
    \item \textit{Atomic interactions} $-$ the energy of atomic interactions (denoted by $H_{a}$) is calculated from the pairwise distances between atoms (atoms that either share bonds, or belong to the same amino acid, are excluded). We use the distance-scaled, finite ideal-gas reference (DFIRE) state potential developed by \cite{zhou2002distance}, which assigns a score to the pairwise distance for each atom type. If a conformation has \textit{steric clashes}, i.e., overlapping atoms too close in space (as defined by a value of 8 in the DFIRE table), we set $H_{a} = \infty$ so that it will have an assigned weight of zero and be discarded.

    \item \textit{Dihedral angles} $-$ the energy of backbone dihedral angles (denoted by $H_\theta$) is calculated from the values of $(\phi_t, \psi_t, \omega_t)$. \cite{wong2017fast} developed empirical distributions $\tilde{p}$ for $(\phi_t, \psi_t, \omega_t)$ based on historical protein data, independently for each amino acid with
    $$
    \tilde{p}(\phi_t, \psi_t, \omega_t)  = \tilde{p}(\phi_t,\psi_t)\tilde{p}(\omega_t),
    $$
    where $\tilde{p}(\phi_t,\psi_t)$ is based on the amino acid type $a_{t}$ and consists of discrete bins of every 5° for each angle (which we assume to be uniformly distributed  within each bin), so each empirical distribution is stored in a 72 by 72 matrix and is visualized similar to Ramachandran plots \citep{Ramachandran1963stereochemistry}; $\tilde{p}(\omega_t)$ is a Gaussian distribution with mean 180° and standard deviation 3° independent of $(\phi_t, \psi_t)$. Following that work, we then set $H_\theta(\mathbf{x}_t) = -\log[ \tilde{p}(\phi_t, \psi_t, \omega_t)]$, which helps guide propagation to more likely regions in the dihedral space.
    
    \item \textit{Geometric (or closure) constraints} $-$ an indicator (denoted by $\mathcal{I}_t$) for the feasibility of the distance between the current amino acid $\mathbf{x}_t$ and the $\text{C}_{\alpha}$ atom after the end of the segment. $\mathcal{I}_t=0$ indicates that it is geometrically infeasible for the segment to form a continuous backbone structure with the rest of the protein, and should be discarded. The allowable range of distances for ${\text{C}}_{t}$ to ${\text{C}_\alpha}_{T+2}$ and ${\text{C}_\alpha}_{t+1}$ to ${\text{C}_\alpha}_{T+2}$ where we set $\mathcal{I}_t=1$ are obtained from \cite{wong2017fast}.
\end{itemize}

Combining the three components, the incremental energy of $\mathbf{x}_{t}$ for $t>0$ is defined as
\begin{equation}
    {H}(\mathbf{x}_{t} \mid \mathbf{x}_{0:t-1}) =\left\{\begin{array}{ll} H_{a}(\mathbf{x}_{t} \mid \mathbf{x}_{0:t-1})/10 + H_\theta(\mathbf{x}_t)  & \mbox{if }\mathcal{I}_t=1 \\ \\ \infty & \mbox{if }\mathcal{I}_t=0\end{array}\right.
    \label{eq:proposed_energy_function}
\end{equation}
where $H_{a}(\mathbf{x}_{t} \mid \mathbf{x}_{0:t-1})$ is the energy for the atomic interactions of $\text{C}_{t}$, $\text{O}_{t}$, $\text{N}_{t+1}$, and ${\text{C}_\alpha}_{t+1}$ corresponding to the sampled $\mathbf{x}_{t}$, and $1/10$ is the relative weight assigned to the $H_{a}$ component following \cite{wong2018exploring}. The auxiliary distribution $p_t$ is then
\begin{align*}
    p_t(\mathbf{x}_{0: t}) &\propto \exp{(-H_{a}(\mathbf{x}_{0})/10)}\tilde{p}(\mathbf{x}_{0})\mathcal{I}_0\prod_{s=1}^t \exp{(-H_{a}(\mathbf{x}_{s} \mid \mathbf{x}_{0:s-1})/10)}\tilde{p}(\mathbf{x}_{s})\mathcal{I}_s.
\end{align*}

It is convenient to set $\eta(\mathbf{x}_{t}\mid \mathbf{x}_{0:t-1}) = \tilde{p}(\mathbf{x}_{t})$ in the proposed framework: to draw $\mathbf{x}_{t}$ from $\tilde{p}(\mathbf{x}_{t})$, we sample a bin for $(\phi, \psi)$ from its empirical distribution then draw a $(\phi_t, \psi_t)$ value uniformly from the selected bin, and $\omega_t$ is drawn from its normal distribution.
%To apply the proposed SMC framework, we set $\eta(\mathbf{x}_{t}\mid \mathbf{x}_{0:t-1}) = \tilde{p}(\mathbf{x}_{t})$ and thus
Equation \eqref{eq:upsampled_weight} for upsampled weights thus simplifies to
\begin{align}
    w(\mathbf{x}_{0:t}) &\propto w(\mathbf{x}_{0:t-1})  \frac{p_t(\mathbf{x}_{0:t})}{p_{t-1}(\mathbf{x}_{0: t-1}) \eta(\mathbf{x}_{t} \mid \mathbf{x}_{0:t-1})}\notag \\&= w(\mathbf{x}_{0:t-1})  \exp{(- H_{a}(\mathbf{x}_{t} \mid \mathbf{x}_{0:t-1})/10)}  \mathcal{I}_t.
    \label{eq:special_increment}
\end{align}

\section{Simulation studies}\label{section: simus}

\subsection{Simulation I: SMC performance for different values of  $M$}\label{section: simu_one}

The upsample size $M$ plays an important role in the upsampling-downsampling framework.
For a finite space $\mathcal{X}$, a natural choice for $M$ can be to take $M=|\mathcal{X}|$ \citep[e.g.,][]{fearnhead2003line}. However, there may not be an intuitive choice for $M$ for a continuous space $\mathcal{X}$. Since the computational complexity of our SMC method is $O(MN)$, we investigate how $M$ influences the performance of Algorithm \ref{our_SMC} for fixed values of $MN$, by using a protein segment and its structural quantities as an illustrative example. 

In this simulation, we considered sampling the length 10 segment at amino acid positions 282--291 of the protein with PDB ID \textit{1ds1A}, i.e., the atoms $\text{C}_{282}, \text{O}_{282}, \text{N}_{283}, {\text{C}_\alpha}_{283}, \dots$, $\text{C}_{291}, \text{O}_{291}, \text{N}_{292}, {\text{C}_\alpha}_{292}$ with the rest of the protein held fixed. As protein structural quantities (see Section \ref{subsec:structquants}), we considered estimating the Boltzmann average for the number of atomic contacts of ${\text{C}_\alpha}$ at each of positions 283--292, which we denote by $n({\text{C}_\alpha}_{283}), \ldots, n({\text{C}_\alpha}_{292})$. The values of $M$ and $N$ used for the simulation were $MN \in \{10^5, 5 \times 10^5, 10^6\}$ and $M \in \{1, 5, 10, 20, 50, 100, 200, 500, 1000\}$. For each combination of $M$ and $N$, we ran 100 independent repetitions of our SMC method.

For $M \in \{1, 200, 500, 1000\}$, we found that some of the 100 repetitions ended prematurely due to no particles having positive weights (i.e., $H_a = \infty$ or $\mathcal{I}_t=0$ for all propagated particles, see Section \ref{subsec:energy_function}). Note that taking $M=1$ is practically equivalent to SISR in this example, since case (ii) resampling was needed at each $t$. This result indicates that SISR is not a suitable algorithm for the protein sampling application; $M=1$ produced many duplicates among the $N$ particles at each $t$, and this insufficient exploration of the target distribution could eventually lead to all dead ends. For large $M$ (200, 500, 1000), the corresponding $N$ was too small to preserve a sufficient number of particles, which likewise could lead to all dead ends. While increasing $MN$ could potentially avoid this issue, doing so is computationally expensive.

For $M \in \{5, 10, 20, 50, 100\}$, all SMC repetitions could reliably finish and we computed the variance of the 100 estimates for each quantity. The results are presented in Table \ref{table:comparison_MN}. We observe that the variances of the estimates for a given value of $M$ decrease as $MN$ increases, as expected. Furthermore, for each fixed value of $MN$, the variances tend to decrease as $M$ increases from 5 to 20, and then increase as $M$ increases from 20 to 100. These results suggest that taking $M=20$ is a good choice for the protein sampling problem, which we use in the subsequent simulation and real data example.

\begin{table}[t!]
    \caption{Variances of the SMC estimates over 100 repetitions, when estimating the Boltzmann average of the number of atomic contacts of ${\text{C}_\alpha}$ at each position from 283 to 292, for different combinations of $MN$ and $M$.}
\resizebox{\linewidth}{!}{\label{table:comparison_MN}\par
\centering
    \begin{tabular}{cccccccccccc}
    \hline
     & $M$ & $n(\text{C}_\alpha{}_{283})$ & $n(\text{C}_\alpha{}_{284})$ & $n(\text{C}_\alpha{}_{285})$ & $n(\text{C}_\alpha{}_{286})$ & $n(\text{C}_\alpha{}_{287})$ & $n(\text{C}_\alpha{}_{288})$ & $n(\text{C}_\alpha{}_{289})$ & $n(\text{C}_\alpha{}_{290})$ & $n(\text{C}_\alpha{}_{291})$ & $n(\text{C}_\alpha{}_{292})$\\ \hline
    $MN=100000$ & 5 & 0.277 & 1.557 & 3.547 & 10.004 & 3.525 & 3.467 & 31.481 & 8.299 & 4.119 & 0.599\\ 
        & 10 & 0.160 & 0.959 & 3.378 & 10.335 & 2.525 & 4.194 & 25.199 & 10.376 & 3.843 & 0.680\\
        
        & 20 & 0.181 & 0.814 & 2.695 & 7.421 & 2.633 & 2.269 & 12.973 & 9.321 & 3.119 & 0.453\\ 
        & 50 & 0.313 & 1.908 & 4.237 & 10.802 & 3.890 & 3.157 & 19.476 & 10.196 & 4.070 & 0.576\\ 
        & 100 & 0.428 & 3.117 & 7.379 & 20.195 & 6.050 & 6.695 & 26.506 & 14.656 & 5.471 & 0.792\\ 
         \hline
        
        $MN=500000$ & 5 & 0.083 & 0.671 & 1.198 & 5.014 & 1.179 & 1.610 & 13.895 & 3.041 & 1.417 & 0.346\\ 
        & 10 & 0.063 & 0.560 & 1.183 & 3.489 & 1.167 & 1.421 & 9.391 & 3.400 & 1.443 & 0.258\\ 
        & 20 & 0.034 & 0.202 & 0.578 & 2.042 & 0.696 & 0.749 & 3.629 & 2.828 & 0.856 & 0.113\\ 
        & 50 & 0.061 & 0.345 & 1.379 & 2.542 & 0.945 & 0.799 & 4.270 & 2.824 & 1.105 & 0.136\\
        & 100 & 0.139 & 0.697 & 1.812 & 5.946 & 1.696 & 2.192 & 8.269 & 5.148 & 1.355 & 0.228\\ 
         \hline
        
        $MN=1000000$ & 5 & 0.051 & 0.374 & 0.637 & 2.504 & 0.523 & 1.058 & 6.682 & 1.914 & 0.590 & 0.096\\ 
         & 10 & 0.035 & 0.234 & 0.499 & 1.741 & 0.505 & 0.860 & 5.571 & 1.258 & 0.527 & 0.147\\ 
        & 20 & 0.024 & 0.132 & 0.556 & 0.762 & 0.315 & 0.226 & 1.222 & 1.327 & 0.388 & 0.070\\
        
        & 50 & 0.036 & 0.179 & 0.535 & 1.421 & 0.521 & 0.482 & 2.498 & 1.504 & 0.512 & 0.065\\
        
        & 100 & 0.058 & 0.369 & 1.070 & 2.529 & 0.744 & 1.088 & 3.364 & 2.319 & 0.857 & 0.129\\

        \hline
    \end{tabular}}
\end{table}

\subsection{Simulation II: SMC convergence and comparison with importance sampling}\label{section: simu_two}

As a second simulation study, we show numerically that the Boltzmann averages estimated by our SMC method converge as $N$ increases. In doing so, we also illustrate the inefficiency of naive importance sampling for the protein application. We sample the length 4 segment at amino acid positions 85--88 of the protein with PDB ID \textit{1cruA}, i.e., the atoms $\text{C}_{85}, \text{O}_{85}, \text{N}_{86}, {\text{C}_\alpha}_{86}, \dots, \text{C}_{88}, \text{O}_{88}, \text{N}_{89}, {\text{C}_\alpha}_{89}$. As protein structural quantities, we considered: the distance between ${\text{C}_\alpha}_{86}$ and ${\text{C}_\alpha}_{89}$, denoted by $d({\text{C}_\alpha}_{86}$,${\text{C}_\alpha}_{89})$; the number of atomic contacts of ${\text{C}_\alpha}$ at each of positions 86--89.

By selecting a short (length 4) segment, we can obtain a good approximation to the ground truth by importance sampling. To do so, we %draw $10^6$ 
repeatedly draw samples of the entire vector $\mathbf{x}_{0:3}$ from the importance distributions $\eta(\mathbf{x}_0), \eta(\mathbf{x}_1\mid \mathbf{x}_0),\dots, \eta(\mathbf{x}_3\mid \mathbf{x}_{0:2})$, i.e., the empirical distributions for the dihedral angles (Section \ref{subsec:energy_function}), and evaluate $H(\mathbf{x}_{0:3})$ to obtain the importance weights. Even for our short segment, this naive method is very inefficient: only $\sim$0.6\% of samples satisfied the geometric constraint ($\mathcal{I}_t=1$ for each $t$); of those, $\sim$8.6\% were free of steric clashes ($H_a < \infty$). Thus, effectively $\sim$2 million draws from $\eta$ were needed to produce 1000 valid samples of $\mathbf{x}_{0:3}$ with positive weights in this case, for an overall rejection rate of $99.95\%$. The probability of obtaining a valid sample with importance sampling would further decrease with longer segments, so would not be a computationally feasible way to estimate Boltzmann averages in general.

Table \ref{tab:benchmark} displays the ground truth for the Boltzmann averages of the five quantities of interest, as approximated from a total of $10^6$ samples (with positive weights) generated using the naive importance sampling approach. To verify that $10^6$ samples provide a good proxy for the ground truth here, we randomly divided them into two subgroups of 500,000 samples, repeating 50 times to form a total of 100 groups of 500,000 samples. The same quantities were estimated from each of the 100 groups. For each quantity, the standard deviations of these estimates was less than $0.05\%$ of the Boltzmann average computed from the full $(10^6)$ sample, which indicates the stability desired.

\begin{table}[t!]
\caption{Boltzmann averages of the five quantities of interest. The $10^6$ samples obtained from importance sampling were used as a proxy for the ground truth.} 
\label{tab:benchmark}\par
\centering
\begin{tabular}{ |p{2.5cm}|p{1.5cm}|p{1.5cm}|p{1.5cm}|p{1.5cm}|}
 \hline
  $d(\text{C}_\alpha{}_{86}$,$\text{C}_\alpha{}_{89})$ & $n(\text{C}_\alpha{}_{86})$ & $n(\text{C}_\alpha{}_{87})$ & $n(\text{C}_\alpha{}_{88})$ & $n(\text{C}_\alpha{}_{89})$\\
 \hline
  8.850 & 29.420 & 39.862 & 56.252 & 37.742\\
 \hline
\end{tabular}
\end{table}

Having established a proxy for the ground truth, we proceed to run our SMC method to estimate these same quantities. We used $M=20$ (as chosen in Section \ref{section: simu_one}) and a range of $N$ values from 1000 to 500,000, running 100 repetitions for each combination of $M$ and $N$. The RMSEs to the ground truth based on these SMC repetitions are shown in Table \ref{table:RMSE_of_SMC_20M}. The RMSEs can be seen to steadily decrease as the number of particles $N$ increases, reaching to $<0.015\%$ of the ground truth values at $N=500000$.

\begin{table}[t!]
\caption{RMSEs of the five quantities, based on 100 repetitions of  our SMC method with $M=20$ and different values of $N$.}\label{table:RMSE_of_SMC_20M}\par
\centering
\begin{tabular}{ |p{1.5cm}|p{2.5cm}|p{1.5cm}|p{1.5cm}|p{1.5cm}|p{1.5cm}|}
 %\hline
 %\multicolumn{6}{|c|}{RMSEs of our SMC method} \\
 \hline
 $N=$ & $d(\text{C}_\alpha{}_{86}$,$\text{C}_\alpha{}_{89})$ & $n(\text{C}_\alpha{}_{86})$ & $n(\text{C}_\alpha{}_{87})$ & $n(\text{C}_\alpha{}_{88})$ & $n(\text{C}_\alpha{}_{89})$\\
 \hline\hline
 1000 & 0.125 & 0.760 & 1.102 & 1.297 & 0.777\\
 2000 & 0.085 & 0.501 & 0.697 & 1.031 & 0.524\\
 5000 & 0.058 & 0.354 & 0.468 & 0.639 & 0.335\\
 10000 & 0.043 & 0.257 & 0.300 & 0.466 & 0.240\\
 20000 & 0.030 & 0.219 & 0.237 & 0.289 & 0.159\\
 50000 & 0.020 & 0.121 & 0.129 & 0.206 & 0.112\\
 100000 & 0.014 & 0.085 & 0.105 & 0.143 & 0.080\\
 200000 & 0.010 & 0.065 & 0.068 & 0.010 & 0.055\\
 500000 & 0.007 & 0.039 & 0.051 & 0.063 & 0.043\\
 \hline
\end{tabular}
\end{table}

Finally, we can quantify the gains in computational efficiency of our SMC method relative to  importance sampling. We generated 500, 1000, and 2000 valid samples (with positive weights) using the naive method, running 100 repetitions for each to estimate the same quantities. The RMSEs to the ground truth based on these importance sampling repetitions are shown in Table \ref{table:RMSE_of_Random_Sampling}. Comparing Tables \ref{table:RMSE_of_SMC_20M} and \ref{table:RMSE_of_Random_Sampling}, we see that the RMSEs for SMC with $N=10000$ and $M=20$ are smaller than for importance sampling with 2000 valid samples. On a single core of a Xeon Gold 6230 2.1 GHz processor, the former had a time cost of 140 seconds per repetition, while the latter required about 7 hours per repetition. We note that the gains from SMC would be even more pronounced for realistic longer segments, where importance sampling would suffer from increasingly high rejection rates.

\begin{table}[t!]
\caption{RMSEs of the five quantities, based on 100 repetitions of importance sampling with different sample sizes.}
\label{table:RMSE_of_Random_Sampling}\par
\centering
\begin{tabular}{ |p{1.5cm}|p{2.5cm}|p{1.5cm}|p{1.5cm}|p{1.5cm}|p{1.5cm}|}
 %\hline
 %\multicolumn{6}{|c|}{RMSEs of importance sampling} \\
 \hline
 Samples & $d(\text{C}_\alpha{}_{86}$,$\text{C}_\alpha{}_{89})$ & $n(\text{C}_\alpha{}_{86})$ & $n(\text{C}_\alpha{}_{87})$ & $n(\text{C}_\alpha{}_{88})$ & $n(\text{C}_\alpha{}_{89})$\\
 \hline\hline
 500 & 0.095 & 0.620 & 0.789 & 1.267 & 0.598\\
 1000 & 0.086 & 0.475 & 0.520 & 0.764 & 0.441\\
 2000 & 0.058 & 0.308 & 0.444 & 0.634 & 0.362\\
 \hline
\end{tabular}
\end{table}

\section{Example: Estimating atomic contacts in the SARS-CoV-2 spike protein}\label{section: application}

The COVID-19 pandemic was caused by the novel coronavirus SARS-CoV-2, with the first identified outbreak in Wuhan, China, in 2019 \citep{chen2020emerging}. This virus latches onto a human host cell via an interaction between its spike protein's receptor-binding domain (RBD) and the host cell's angiotensin-converting enzyme 2 (ACE2) receptor \citep{lan2020structure}. The conformational dynamics and flexibility of the RBD binding interface, i.e., the segments of the RBD involved in binding, has thus been of scientific interest towards the development of therapeutics. %and worth studying for a better understanding of the designs and improvement of the medical treatments and vaccines.
Four such segments of amino acids have been identified to compose the RBD--ACE2 binding interface, with positions 472--490 (known as Loop 3) and 495--506 (known as Loop 4) having the most dynamic movement due to their conformational flexibility
\citep{dehury2020effect,nguyen2020does,ali2020dynamics}.

Since the initial outbreak, laboratory work \citep[e.g.,][]{wrapp2020cryo} has provided PDB structures of the SARS-CoV-2 spike protein. These static snapshots of the spike protein have been useful as a starting point for the further study of its conformational dynamics; e.g., \cite{williams2021molecular} used PDB structures of the RBD in its prefusion state (i.e., prior to binding with ACE2) to study Loop 3. Their molecular dynamics simulations revealed that Loop 3 is the most flexible of the RBD segments, 
with the potential to adopt conformations that are notably different from the PDB structure. These findings could open avenues for drug design to block the RBD--ACE2 interaction by targeting these conformations. Both the Delta and Omicron variants of SARS-CoV-2 which became widespread had key mutations within Loop 3, which further reinforced its importance for scientists tackling COVID-19.

For this example, we took the same starting PDB structure as \cite{williams2021molecular} and applied our SMC method to sample backbone conformations of Loop 3. We then estimated the Boltzmann averages for the number of atomic contacts for each $\text{C}_\alpha$ in the segment, which may be interpreted as their relative degree of protrusion. Amino acid positions that have fewer contacts on average tend to be more exposed to the protein surface, and hence potentially more involved in binding activities. We ran our SMC method with $M=20$ (as in Section \ref{section: simu_one}) and $N = 500000$, and the estimated Boltzmann averages for $n({\text{C}_\alpha}_{473}), \ldots, n({\text{C}_\alpha}_{489})$ are plotted in Figure \ref{fig:sarscov2}. These results suggest that on average, position 478 is the most exposed amino acid (among several consecutive positions with high protrusion, namely 476--479), while position 483 protrudes the least. Interestingly, \cite{williams2021molecular} note that 476, 477, 478, and 483 were the top four positions were mutations were most commonly observed, and that the conformational flexibility of Loop 3 appeared to be resilient to single mutations at these positions.

\begin{figure}[h]
\centering
\includegraphics[width=13cm]{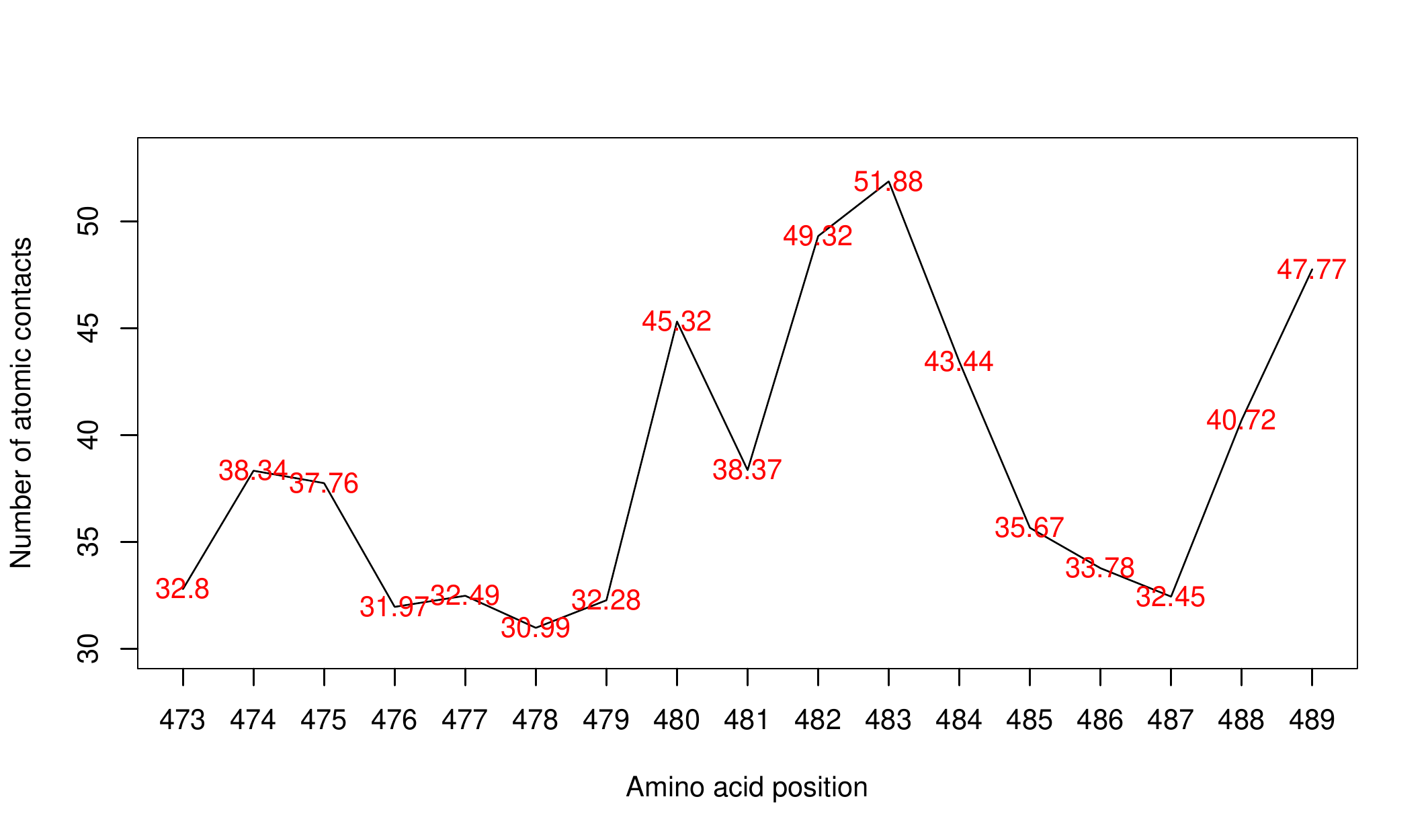}
\caption{Estimated Boltzmann averages for the number of atomic contacts of $\text{C}_\alpha$ atoms at amino acid positions 473 to 489 of the SARS-CoV-2 spike protein, using the samples obtained from running our SMC method with $N=500000$ and $M=20$.}
\label{fig:sarscov2}
\end{figure}

\section{Conclusion and Discussion}\label{section: conclusion}

In this paper, we proposed an SMC method that features an upsampling-downsampling framework, with emphasis on sampling backbone conformations of protein segments from the Boltzmann distribution. Previous methods for protein structure prediction were usually designed to search for the conformations with the lowest energy, without regard for properly weighted samples. In contrast, our SMC samples can be used to estimate Boltzmann averages over the full range of dynamic movement, according to a given energy function. Furthermore, previous SMC methods (such as SISR) were inapplicable for this protein sampling problem due to the particle degeneracy encountered, which highlights the usefulness of our approach. We showed the theoretical validity of the framework, and demonstrated its performance via simulation studies and an illustrative example using a key segment of the SARS-CoV-2 spike protein.

Compared to SISR with a given particle size $N$, our SMC method can provide more accurate estimates of Boltzmann averages but comes with a price: with upsample size $M$ its computational complexity is $O(MN)$ compared to $O(N)$ for SISR.  Thus, choosing a value of $M$ involves a tradeoff between computing speed and exploration effort of high-density regions. As illustrated in Section \ref{section: simu_one}, it is sensible to fix the computational budget $MN$ before selecting $M$ for a given application. While there is no universally optimal choice of $M$ for all SMC applications, we can provide some intuition. When the target distribution is more uniform, a larger $N$ is preferred; in the extreme case where the target distribution is uniform over the support, the optimal choice is clearly $M=1$ to maximize the number of samples $N$ for a fixed $MN$. In contrast when the target distribution has many sharp local modes or is highly constrained, % with significantly higher densities,
a larger $M$ is preferred to ensure that SMC sufficiently explores the regions with positive density during propagation steps, and a specific value of $M$ could be chosen according to Monte Carlo variance. Therefore, while the upsampling-downsampling SMC framework is generally applicable, its efficacy will depend on the target distribution of interest. When sampling protein segments, steric clashes and geometric constraints lead to a rough energy surface with many local modes. This situation is well-suited for our upsampling-downsampling framework, whereas other SMC methods could fail.

Potential future applied work could involve protein analyses with different structural quantities or energy functions. It is straightforward to adapt our SMC method for this purpose. The sampling framework could also be extended to consider both the backbone and the side chains. One approach would be to sample the side chains after sampling backbone conformations, effectively applying SMC twice and accumulating the energy contributions.

\section*{Acknowledgements}

We thank Martin Lysy and Glen McGee for constructive comments on the manuscript. This work was partially supported by Discovery Grant RGPIN-2019-04771 from the Natural Sciences and Engineering Research Council of Canada.
\par

%%%%%%%%%%%%%%%%%%%%%%%%%%%%%%%%%%%%%%%%%%%%%%%%%%%%%%%%%%%%%%%%%%%%%%%%%%%%%%%%%%%%%%%%%%

%%%%%%%%%%%%%%%%%%%%%%%%%%%%%%%%%%%%%%%%%%%%%%%%%%%%%%%%%%%%%%%%%%%%%%%%%%%%%%%%%%%%%%%%%%%%%%%%%%%%%%%%%%%%%%%%%%%%%%%%%%%%
\section*{Supplementary Materials}
The Supplementary Material contains the proofs of Theorems 1 and 2 in Section \ref{sec:proposed:methodology}.

\par
%%%%%%%%%%%%%%%%%%%%%%%%%%%%%%%%%%%%%%%%%%%%%%%%%%%%%%%%%%%%%%%%%%%%%%%%%%%%%%%%%%%%%%%%%%%%%%%%%%%%%%%%%%%%%%%%%%%%%%%%%%%%

\bibhang=1.7pc
\bibsep=2pt
\fontsize{9}{14pt plus.8pt minus .6pt}\selectfont
\renewcommand\bibname{\large \bf References}
%\begin{thebibliography}{11}
\expandafter\ifx\csname
natexlab\endcsname\relax\def\natexlab#1{#1}\fi
\expandafter\ifx\csname url\endcsname\relax
  \def\url#1{\texttt{#1}}\fi
\expandafter\ifx\csname urlprefix\endcsname\relax\def\urlprefix{URL}\fi

%% use bibfile 
\bibliographystyle{chicago}      % Chicago style, author-year citations
\bibliography{reference}   % name your BibTeX data base

\newpage

\appendix

\section{Proper weighting of the proposed sampling scheme}\label{proof:A}
Recall $p_t(\mathbf{x}_{0:t})$ denotes the auxiliary distribution for step $t$ in Section 3.2 in this context. Let 
$\mathcal{S}_0 = \{\mathbf{x}_0^{(n_0)}; n_0=1,\dots, NM\}$ and $\mathcal{S}_t = \mathcal{S}_{t-1} \cup \{(\mathbf{x}_0^{(n_0)},\dots, \mathbf{x}_t^{(n_t)}); n_0=1,\dots, NM; n_1 = 1,\dots, M;\dots; n_t = 1,\dots, M\}$, i.e., $\mathcal{S}_t$ is the collection of upsampled particles up to time $t$. Any $\mathbf{x}_{0:t}^* = (\mathbf{x}_0^{(n_0^*)}, \dots, \mathbf{x}_t^{(n_t^*)}) \in \mathcal{S}_t$ satisfies $\mathbf{x}_{0:s}^* = (\mathbf{x}_0^{(n_0^*)}, \dots, \mathbf{x}_s^{(n_s^*)}) \in \mathcal{S}_t$ for $s\leq t$. Let $\eta(\mathbf{x}_{0:t})$ denote $\eta(\mathbf{x}_0)\prod_{s=1}^t \eta(\mathbf{x}_{s}  \mid \eta(\mathbf{x}_{0:s-1})$ and $\mathcal{Q}_{0:t}$ denote the $\sigma$-algebra generated by $Q(\mathbf{x}_{0:t})$ conditional on $\mathcal{S}_t$. 

When $t=0$, it is obvious that $E_{\mathcal{Q}_0}\left[ Q(\mathbf{x}_{0}^*) \mid \mathcal{S}_0 \right] \propto \frac{p_{0}(\mathbf{x}_{0}^*)}{\eta(\mathbf{x}_{0}^*)}$ and for any square integrable function $h_0(\mathbf{x}_{0})$,
\begin{align*}
    E_\eta\left\{ E_{\mathcal{Q}_0}\left[ h_0(\mathbf{x}_{0}^*)Q(\mathbf{x}_{0}^*) \mid \mathcal{S}_0 \right] \right\} &= E_\eta\left\{ h_0(\mathbf{x}_{0}^*) E_{\mathcal{Q}_0}\left[ Q(\mathbf{x}_{0}^*) \mid \mathcal{S}_0 \right] \right\}\\
    &\propto E_\eta\left\{ h_0(\mathbf{x}_{0}^*) \frac{p_{0}(\mathbf{x}_{0}^*)}{\eta(\mathbf{x}_{0}^*)} \right\}\\
    &= \int h_0(\mathbf{x}_{0}^*) \frac{p_{0}(\mathbf{x}_{0}^*)}{\eta(\mathbf{x}_{0}^*)} \eta(\mathbf{x}_{0}^*) d\mathbf{x}_{0}^*\\
    &= E_{p_0}\left[ h(\mathbf{x}_{0}) \right],
\end{align*}
which justifies the proper weighting condition.

We now proceed by induction: assume $E_{\mathcal{Q}_{0:t-1}}\left[ Q(\mathbf{x}_{0:t-1}^*) \mid \mathcal{S}_{t-1} \right] \propto \frac{p_{t-1}(\mathbf{x}_{0:t-1}^*)}{\eta(\mathbf{x}_{0:t-1}^*)}$ holds.
Conditional on $\mathcal{S}_{t}$ and $Q(\mathbf{x}_{0:t-1})$ for all $\mathbf{x}_{0:t-1} \in \mathcal{S}_{t-1}$, our SMC produces $Q(\mathbf{x}_{0:t}^*)$ as
$$Q(\mathbf{x}_{0:t}^*)  \propto\left\{\begin{array}{ll}
\frac{Q(\mathbf{x}_{0:t-1}^*)p_{t}(\mathbf{x}_{0:t}^*)}{p_{t-1}(\mathbf{x}_{0:t-1}^*)\eta(\mathbf{x}_{t}^* \mid \mathbf{x}_{0:t-1}^*)}\frac{1}{q(\mathbf{x}_{0:t}^*)} & \text { with probability } q(\mathbf{x}_{0:t}^*) \\ \\ 0 & \text { otherwise }\end{array}\right.$$ where $q(\mathbf{x}_{0:t}^*) = \min\{c_t \frac{Q(\mathbf{x}_{0:t-1}^*)p_{t}(\mathbf{x}_{0:t}^*)}{p_{t-1}(\mathbf{x}_{0:t-1}^*)\eta(\mathbf{x}_{t}^* \mid \mathbf{x}_{0:t-1}^*)},1\}$ with $c_t$ being the root of $$
\sum_{\mathbf{x}_{0:t} \in \mathcal{S}_t} \min \left\{c_t \frac{Q(\mathbf{x}_{0:t-1})p_{t}(\mathbf{x}_{0:t})}{p_{t-1}(\mathbf{x}_{0:t-1})\eta(\mathbf{x}_{t} \mid \mathbf{x}_{0:t-1})}, 1\right\}=N.
$$
We can rewrite $Q(\mathbf{x}_{0:t}^*)$ as $$Q(\mathbf{x}_{0:t}^*) \propto \frac{Q(\mathbf{x}_{0:t-1}^*)p_{t}(\mathbf{x}_{0:t}^*)}{p_{t-1}(\mathbf{x}_{0:t-1}^*)\eta(\mathbf{x}_{t}^* \mid \mathbf{x}_{0:t-1}^*)}\frac{1}{q(\mathbf{x}_{0:t}^*)} \mathbf{I}(\mathbf{x}_{0:t}^*),$$ where $\mathbf{I}(\mathbf{x}_{0:t}^*)$ denotes an indicator function defined conditional on $\mathcal{S}_t$ and $Q(\mathbf{x}_{0:t-1}^*)$ with 
$
E_\mathbf{I}[\mathbf{I}(\mathbf{x}_{0:t}^*) \mid Q(\mathbf{x}_{0:t-1}^*), \mathcal{S}_t] = q(\mathbf{x}_{0:t}^*).
$ Note that $\mathbf{I}(\mathbf{x}_{0:t}^*)$ is independent of future descendants, so for $0< r<s\leq t$, 
$$
E_\mathbf{I}\left[ \mathbf{I}(\mathbf{x}_{0:r}^*) \mid Q(\mathbf{x}_{0:r-1}^*),\mathcal{S}_{s} \right] = E_\mathbf{I}\left[ \mathbf{I}(\mathbf{x}_{0:r}^*) \mid Q(\mathbf{x}_{0:r-1}^*), \mathcal{S}_{r} \right].
$$

Now we can see that $Q(\mathbf{x}_{0:t}^*)$ is proportional to the product of the two random variables $Q(\mathbf{x}_{0:t-1}^*)$ and $\mathbf{I}(\mathbf{x}_{0:t}^*)$, with $\mathbf{I}(\mathbf{x}_{0:t}^*)$ conditional on $Q(\mathbf{x}_{0:t-1}^*)$. By the tower rule, we have
\begin{align*}
    E_{\mathcal{Q}_{0:t}}\left\{ Q(\mathbf{x}_{0:t}^*) \mid \mathcal{S}_{t} \right\} &\propto E_{\mathcal{Q}_{0:t-1}}\left\{ \frac{Q(\mathbf{x}_{0:t-1}^*)p_{t}(\mathbf{x}_{0:t}^*)}{p_{t-1}(\mathbf{x}_{0:t-1}^*)\eta(\mathbf{x}_{t}^* \mid \mathbf{x}_{0:t-1}^*)}\frac{1}{q(\mathbf{x}_{0:t}^*)} E_{\mathbf{I}}\left[ \mathbf{I}(\mathbf{x}_{0:t}^*) \mid Q(\mathbf{x}_{0:t-1}^*), \mathcal{S}_{t} \right] \mid \mathcal{S}_{t} \right\} \\
    &= E_{\mathcal{Q}_{0:t-1}}\left[\frac{Q(\mathbf{x}_{0:t-1}^*)p_{t}(\mathbf{x}_{0:t}^*)}{p_{t-1}(\mathbf{x}_{0:t-1}^*)\eta(\mathbf{x}_{t}^* \mid \mathbf{x}_{0:t-1}^*)} \mid \mathcal{S}_{t} \right]\\
    &= E_{\mathcal{Q}_{0:t-1}}\left[ Q(\mathbf{x}_{0:t-1}^*) \mid \mathcal{S}_{t} \right] \frac{p_{t}(\mathbf{x}_{0:t}^*)}{p_{t-1}(\mathbf{x}_{0:t-1}^*)\eta(\mathbf{x}_{t}^* \mid \mathbf{x}_{0:t-1}^*)}.
\end{align*}
Note that $Q(\mathbf{x}_{0:t-1}^*)$ is only conditional on $\mathcal{S}_{t-1}$ and thus independent of future descendants, so 
$$
E_{\mathcal{Q}_{0:t-1}}\left[ Q(\mathbf{x}_{0:t-1}^*) \mid \mathcal{S}_{t} \right] = E_{\mathcal{Q}_{0:t-1}}\left[ Q(\mathbf{x}_{0:t-1}^*) \mid \mathcal{S}_{t-1} \right]
$$
and thus
$$
E_{\mathcal{Q}_{0:t}}\left\{ Q(\mathbf{x}_{0:t}^*) \mid \mathcal{S}_{t} \right\} \propto E_{\mathcal{Q}_{0:t-1}}\left[ Q(\mathbf{x}_{0:t-1}^*) \mid \mathcal{S}_{t-1} \right] \frac{p_{t}(\mathbf{x}_{0:t}^*)}{p_{t-1}(\mathbf{x}_{0:t-1}^*)\eta(\mathbf{x}_{t}^* \mid \mathbf{x}_{0:t-1}^*)} \propto \frac{p_{t}(\mathbf{x}_{0:t}^*)}{\eta(\mathbf{x}_{0:t}^*)}.
$$
Therefore, for any square integrable function $h_t(\mathbf{x}_{0:t})$,
\begin{align*}
    E_\eta\left\{ E_{\mathcal{Q}_{0:t}}\left[h_t(\mathbf{x}_{0:t}^*)Q(\mathbf{x}_{0:t}^*) \mid \mathcal{S}_{t} \right]\right\} &= E_\eta\left\{ h_t(\mathbf{x}_{0:t}^*)E_{\mathcal{Q}_{0:t}}\left[Q(\mathbf{x}_{0:t}^*) \mid \mathcal{S}_{t} \right]\right\}\\
    &\propto E_\eta\left\{h_t(\mathbf{x}_{0:t}^*)\frac{p_{t}(\mathbf{x}_{0:t}^*)}{\eta(\mathbf{x}_{0:t}^*)}\right\}\\
    &= \int h_t(\mathbf{x}_{0:t}^*)\frac{p_{t}(\mathbf{x}_{0:t}^*)}{\eta(\mathbf{x}_{0:t}^*)}\eta(\mathbf{x}_{0:t}^*) d\mathbf{x}_{0:t}^*\\
    &= E_{p_t}\left[ h_t(\mathbf{x}_{0:t}) \right],
\end{align*}
which justifies the proper weighting condition.

\section{Minimization of the conditional expected squared error loss}\label{proof:B}
Fearnhead and Clifford (2003) have shown that in the downsampling step, only $N$ of the $Q(\mathbf{x}_{0:t}^{(n,m)})$'s are non-zero so that for 
some random variable $C_t^{(n,m)}$ with $\sigma$-algebra $\mathcal{C}$, we have
$$
Q(\mathbf{x}_{0:t}^{(n,m)})=\left\{\begin{array}{ll}C_t^{(n,m)} & \text { with probability } q(\mathbf{x}_{0:t}^{(n,m)}) \\ \\ 0 & \text { otherwise }\end{array}\right.
$$

Assume the set of particles $\{\mathbf{x}^{(n,m)}_{0:t},n=1,\dots,N \text{ and } m=1,\dots,M\}$ are obtained after upsampling and let $\gamma_{t}^{(n,m)} = \left(p_t(\mathbf{x}^{(n,m)}_{0:t})/\eta(\mathbf{x}^{(n,m)}_{0:t})\right)/\sum_{n=1}^N\sum_{m=1}^M\left(p_t(\mathbf{x}^{(n,m)}_{0:t})/\eta(\mathbf{x}^{(n,m)}_{0:t})\right)$, then the expected squared error loss for $Q(\mathbf{x}_{0:t}^{(n,m)})$ can be written as
\begin{align*}
    E_{\mathcal{Q}_{0:t}}\left[(Q(\mathbf{x}_{0:t}^{(n,m)})-\gamma_{t}^{(n,m)})^2 \mid \mathcal{S}_{t}\right]
     &= E_\mathcal{C}\left\{E_{\mathcal{Q}_{0:t}}\left[(Q(\mathbf{x}_{0:t}^{(n,m)})-\gamma_{t}^{(n,m)})^2\mid \mathcal{S}_{t}, C^{(n,m)}\right] \middle\vert\ \mathcal{S}_{t} \right\} \\
    &= q(\mathbf{x}_{0:t}^{(n,m)})\{E_\mathcal{C}(C^{(n,m)}-\gamma_{t}^{(n,m)} \mid \mathcal{S}_{t} )\}^2\\
    &+q(\mathbf{x}_{0:t}^{(n,m)})\text{Var}(C^{(n,m)} \mid \mathcal{S}_{t})+(1-q(\mathbf{x}_{0:t}^{(n,m)})){\gamma_{t}^{(n,m)}}^2,
\end{align*}
which is minimized when $C^{(n,m)}$ is a constant for each $n$ and $m$. To ensure proper weighting as shown in Section A of the supplement, we set $C^{(n,m)} = w(\mathbf{x}_{0:t}^{(n,m)})/q(\mathbf{x}_{0:t}^{(n,m)})$, and then the conditional expected squared error loss is minimized subject to $\sum_{n=1}^N\sum_{m=1}^M q(\mathbf{x}_{0:t}^{(n,m)}) \leq N$, since 
$$\text{min}\left\{E_{\mathcal{Q}_{0:t}}\left[\sum_{n=1}^N\sum_{m=1}^M(Q(\mathbf{x}_{0:t}^{(n,m)})-\gamma_{t}^{(n,m)})^2 \middle\vert\ \mathcal{S}_{t}\right]\right\} = \sum_{n=1}^N\sum_{m=1}^M(w(\mathbf{x}_{0:t}^{(n,m)})/q(\mathbf{x}_{0:t}^{(n,m)}) - \gamma_{t}^{(n,m)})^2.$$

% \vskip .3cm
%\centerline{(Received ???? 20??; accepted ???? 20??)}\par
\end{document}